\documentclass[dvipsnames,letterpaper, 10 pt, journal]{IEEEtran}

\IEEEoverridecommandlockouts	

\usepackage{cite}
\usepackage{amsmath,amssymb,amsfonts,amsthm}
\usepackage{algorithmic}
\usepackage{textcomp}
\usepackage{graphicx,color}
\usepackage{mathrsfs}
\usepackage{tikz}
\usepackage[vlined,ruled]{algorithm2e}
\usepackage{subfigure}
\usepackage{url}
\usepackage{color}
\usepackage{dsfont}
\usepackage{bbm}
\usepackage{booktabs}
\usepackage{array}
\usepackage{yfonts}
\usepackage[normalem]{ulem}
\usepackage{arydshln,leftidx,mathtools}
\usepackage{multicol}
\usepackage{amsmath}
\usepackage{stfloats}
\usepackage{comment}
\usepackage{subfigure}
\usepackage{manyfoot}
\usepackage{enumitem}
\usetikzlibrary{arrows.meta} 

\definecolor{myBlue}{RGB}{49,130,189}

\newtheorem{theorem}{Theorem}[section]

\newtheorem{proposition}[theorem]{Proposition}
\newtheorem{problem}{Problem}

\usepackage{tabularx}
\usepackage{multirow}
\usepackage{makecell}

\newcommand\aamsout{\bgroup\markoverwith{\textcolor{violet}{\rule[0.5ex]{2pt}{1pt}}}\ULon}

\newcommand{\mc}{\mathcal}

\DeclareSymbolFont{bbold}{U}{bbold}{m}{n}
\DeclareSymbolFontAlphabet{\mathbbold}{bbold}

\newcommand{\R}{\mathbb{R}}

\newcommand{\N}{\mathbb{N}}

\newcommand\oprocendsymbol{\hbox{$\square$}}
\newcommand\oprocend{\relax\ifmmode\else\unskip\hfill\fi\oprocendsymbol}

\renewcommand{\theenumi}{(\roman{enumi})}

\newcommand*{\QEDA}{\hfill\ensuremath{\blacksquare}}%

\DeclareNewFootnote{R}[roman]

\graphicspath{{figs/}}

\makeatletter
\let\NAT@parse\undefined
\makeatother
\usepackage[colorlinks,urlcolor=blue, linkcolor=blue, citecolor=blue]{hyperref}

\begin{document}

\title{\LARGE \bf Hybrid-time PINNs: Discontinuity-Aware Physics-Informed Neural Networks for Learning Hybrid Dynamical Systems}

\author{Shivanshu~Tripathi, and Maziar~Raissi %
\thanks{This material is based upon work supported in part by award NSF-CMMI-2618914. S.~Tripathi is with the Department of Electrical and Computer Engineering, and M.~Raissi is with the Department of Mathematics, University of California, Riverside, Riverside, CA 92521, USA.
Emails: \href{mailto:strip008@ucr.edu}{\texttt{strip008@ucr.edu}} and
\href{mailto:maziarr@ucr.edu}{\texttt{maziarr@ucr.edu}}.}
}

\maketitle
\thispagestyle{empty}
\pagestyle{empty}
\suppressfloats[t] 

\begin{abstract}
Hybrid dynamical systems flow continuously but jump at discrete events. Physics-informed neural networks (PINNs) fit data while enforcing known dynamics, but a single network is smooth in time and cannot model these jumps. We show that the more closely a PINN follows a jump, the larger its physics loss must be. We propose hybrid-time PINNs (HT-PINNs), which assign one network to each flow interval of the hybrid time domain and train the networks jointly with the jump times and the unknown parameters. The number of jumps and the sequence of modes are found by enumeration. HT-PINNs recover the number of jumps and the restitution coefficient of a bouncing ball. In a MuJoCo simulation of robotic disassembly, where a robot pulls a press-fitted pin out of a housing through a spring, HT-PINNs find the number of jumps and the sequence of modes, and estimate the static and kinetic friction coefficients with small errors of $0.1\%$ and $1.7\%$.
\end{abstract}

\section{Introduction}
\label{sec:intro}

Hybrid dynamical systems combine continuous flows with instantaneous jumps~\cite{RG-RGS-ART:12}. They arise in robotics, where legged robots make and break contact with the ground~\cite{PMW-MP-YH-AE-NM-ADP:23}; in biology, where fireflies synchronize their flashes~\cite{REM-SHS:90}; and in power systems, where switching events change the network~\cite{IAH-MAP:00}. Between jumps, the state follows a differential equation; when it reaches a guard, a jump resets it and may change the governing differential equation (Fig.~\ref{fig:teaser}).

Models of hybrid systems are rarely known exactly: the modes, guards, and resets often follow from physics, but parameters such as friction and restitution coefficients are uncertain and depend on the environment. In robotic disassembly~\cite{MEA-AR-RS:24}, for example, the friction between press-fitted parts changes with wear and corrosion. A model must therefore combine the known physics with the measurements and reproduce the jumps, whose number, times, and mode sequence are unknown.

Our objective is to model such systems with HT-PINNs. A classical PINN~\cite{MR-PP-GEK:19} represents a trajectory by a network trained to fit the measurements while satisfying the known differential equations, which yields the unmeasured states and the unknown parameters. A PINN, however, is smooth in time and cannot represent a jump. An HT-PINN is instead defined on a hybrid time domain, which indexes the trajectory by the time and the number of jumps. It uses one network per interval between consecutive jumps, trains the networks jointly with the jump times and the parameters, and selects the number of jumps and the mode sequence by enumeration.

\textbf{Related work.} Methods for learning models of dynamical systems are either data-driven, learning the entire model from data, or physics-informed, combining the data with known physics.
Among data-driven methods, hybrid system identification fits switched and piecewise affine models in discrete time~\cite{SP-ALJ-GFT-RV:07,FL-GB:19} or learns the equations of each mode by sparse regression~\cite{NMM-TA-SLB-JNK-JLP:19}, neural event functions~\cite{RTQC-BA-MN:21} learn state-triggered events of neural ordinary differential equations, and neural hybrid automata~\cite{MP-SM-LS-SC-SJO-AY-HA-JP-AG:21} learn the flows of an unknown number of modes and the stochastic jumps between them. Since these methods do not impose the known physics, their models need not obey it. On the other hand, PINNs~\cite{MR-PP-GEK:19,GEK-IGK-LL-PP-SW-LY:21} impose the known physics and have been used as models for control~\cite{EAA-EC-LOS-JPJ-ERDS-JFH:24}, but cannot represent a jump~\cite{MP-MH-MP:21}.
Extensions of PINNs train separate networks on fixed subdomains while keeping the solution continuous across their interfaces~\cite{ADJ-GEK:20,BM-AM-TNM:23}, or capture shocks in partial differential equations~\cite{GL-DE-XM:26}. In the recent work~\cite{BH-CH-MZ-RW:26}, the authors align a PINN with the fault and clearing events of a power system and chain the state across these events. In~\cite{YB-CT-JL-XW-ZZ:26}, the authors use PINN residuals to locate the change points of the parameters of an ordinary differential equation. Differently from~these works, we consider jumps that are triggered by the state through guards and may reset it. For hybrid systems whose flow and jump maps are known up to parameters, the estimator in~\cite{RSJ-SDC-RGS:24} assumes that the parameters enter linearly, the state is measured, and the jumps are detected. In contrast, our work addresses a more general setting, in which the parameters may enter nonlinearly, only an output is sampled, and neither the number nor the times of the jumps are known. However, we do restrict our analysis to hybrid systems whose modes, guards, and resets are known up to parameters.
\begin{figure}[t]
  \centering
  \includegraphics[width=\linewidth]{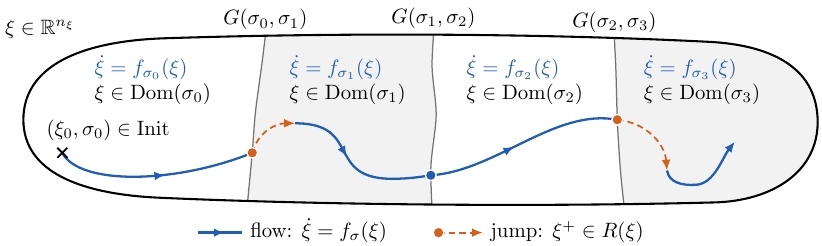}
\caption{Solution of $\mathcal{H}$ flows within a mode and switches modes at guards.}
  \label{fig:teaser}
\end{figure}

\noindent
\textbf{Contributions.} The main contributions of this letter are as follows. First, we design HT-PINNs to model a hybrid system known up to parameters from measurements. This modeling problem is not standard, since the number and times of the jumps and the sequence of modes are unknown. Second, we show that a single smooth network cannot follow a jump closely without a large physics loss. Third and finally, we illustrate the effectiveness of HT-PINNs on a bouncing ball and on a MuJoCo~\cite{ET-TE-YT:12} simulation of pin extraction, where they recover the stick, slip, and release events within $1.1$~ms and the friction coefficients with errors of $0.1\%$ and $1.7\%$.

\noindent
\textbf{Notation.} Let $\R_{\ge0}$ and $\N=\{0,1,\dots\}$ denote the nonnegative reals and integers, respectively. The Euclidean norm and the cardinality of a set $\mc{S}$ are denoted by $|\cdot|$ and $|\mc{S}|$, respectively, and $[a]_+=\max\{a,0\}$. For $\alpha\in\R^n$, $\mathrm{softmax}(\alpha)$ denotes the vector with entries $e^{\alpha_j}/\sum_{i}e^{\alpha_i}$, and $\operatorname{dom}\phi$ denotes the domain of a function~$\phi$.

\section{Hybrid Systems and Problem Formulation}
\label{sec:prelim}


\label{sec:hs}

We model a hybrid system as:
\begin{equation}
\mathcal{H}:\ \begin{cases}
  \dot x \in F(x), & x \in C,\\
  x^+ \in G(x), & x \in D,
\end{cases}
\qquad y = h(x),
\label{eq:hs}
\end{equation}
where $x$ is the state and $y$ is the output. While $x$ is in the {flow set} $C$, it changes continuously according to the {flow map} $F$. When $x$ is in the {jump set} $D$, it can jump instantly to a new value $x^+$ given by the {jump map} $G$. The {output map} $h$ gives the measured output.

The state is $x=(\xi,\sigma)$, where $\xi\in\R^{n_\xi}$ contains the physical variables, such as positions and velocities, and the {mode} $\sigma$ takes values in a finite set $\Sigma\subset\N$ and changes only at jumps. Each mode $\sigma$ has {guard functions} $\gamma_{\sigma,i}(\xi)$, $i\in I(\sigma)$, which decide when the state flows and when it jumps:
\begin{equation}
\begin{aligned}
  C &= \{(\xi,\sigma) : \gamma_{\sigma,i}(\xi) \ge 0 \text{ for all } i \in I(\sigma)\},\\
  D &= \{(\xi,\sigma) : \gamma_{\sigma,i}(\xi) \le 0 \text{ for some } i \in I(\sigma)\}.
\end{aligned}
\label{eq:guards}
\end{equation}
In mode $\sigma$, the physical variables follow the vector field $f_\sigma$. A jump triggered by guard $i$ updates the physical variables with the reset map $R_{\sigma,i}$ and switches the mode to $\nu(\sigma,i)\in\Sigma$. The flow and jump maps are thus
\begin{equation}
\begin{aligned}
F(x) &= \{(f_\sigma(\xi),\,0)\},\\
G(x) &= \{(R_{\sigma,i}(\xi),\,\nu(\sigma,i)) : i\in I(\sigma),\ \gamma_{\sigma,i}(\xi)\le 0\}.
\end{aligned}
\label{eq:modewise}
\end{equation}

{\bf Solution of $\mathcal H$}: A function $\phi=(\xi,\sigma)$ on domain $\cup_{j=0}^{J}[\tau_j,\tau_{j+1}]\times\{j\}\subset \mathbb{R}_{\geq 0} \times \mathbb{N}$, with jump times $0=\tau_0\le\dots\le\tau_{J+1}=T,$  is a solution of $\mathcal{H}$ on $[0,T]$ if
(i)~for each $j$ with $\tau_j<\tau_{j+1}$, $t\mapsto\phi(t,j)$ is continuously differentiable on $[\tau_j,\tau_{j+1}]$, and $\phi(t,j)\in C$ and $\dot\phi(t,j)\in F(\phi(t,j))$ for $t\in(\tau_j,\tau_{j+1})$; and
(ii)~for $j=1,\dots,J$, $\phi(\tau_j,j-1)\in D$ and $\phi(\tau_j,j)\in G(\phi(\tau_j,j-1))$.
Jump~$j$ is triggered by guard~$i$ if $\gamma_{\sigma,i}(\xi)\le0$ and $\phi(\tau_j,j)=(R_{\sigma,i}(\xi),\nu(\sigma,i))$, where $(\xi,\sigma)=\phi(\tau_j,j-1)$.

{\bf Problem setting:} We assume that the modes $\Sigma$, guard indices $I(\sigma)$, transitions $\nu$, and output
map $h$ are known, and $y=h(\xi)$. The vector fields $f_\sigma(\xi;\theta)$, the guard functions
$\gamma_{\sigma,i}(\xi;\theta)$, and the reset maps $R_{\sigma,i}(\xi;\theta)$ are known up
to a parameter $\theta\in\Theta$, which gives a known family
$\{\mathcal H_\theta\}_{\theta\in\Theta}$ of hybrid systems.
The data come from a solution $\phi$ of $\mathcal H_{\theta^\star}$ on
$[0,T]$, with finitely many jumps, and unknown 
$\theta^\star\in\Theta$, whose output is sampled
at times $0\le t_1<\dots<t_K\le T$:
\begin{equation}
y_k = h\big(\xi(t_k,j_k)\big)+\zeta_k,\qquad k=1,\dots,K,
\label{eq:meas}
\end{equation}
where $j_k$ is the number of jumps of $\phi$ up to time $t_k$ and $\zeta_k$ is 
measurement noise. In our setting, we also do not measure the mode and the 
jump times. 

\begin{problem}\label{prob:main}
Given the family $\{\mathcal H_\theta\}_{\theta\in\Theta}$ and the data
$\{(t_k,y_k)\}_{k=1}^{K}$, find a parameter $\theta\in\Theta$ {and} a
solution $\hat\phi=(\hat\xi,\hat\sigma)$ of
$\mathcal H_{\theta}$ on $[0,T]$, represented by neural networks, 
whose output fits the data:
\begin{equation}
h\big(\hat\xi(t_k,\hat j_k)\big)\approx y_k,\qquad k=1,\dots,K,
\label{eq:match}
\end{equation}
where $\hat j_k$ is the number of jumps of $\hat\phi$ up to time $t_k$.
The initial state $\hat\phi(0,0)$, the number of jumps of $\hat\phi$ and the 
jump times are unknown and are found jointly with $\theta$.
\end{problem}

Requiring $\hat\phi$ to be a solution of $\mathcal H_{\theta}$ makes the
neural network physics-informed as it must obey the known physics of the hybrid system,
not just fit the data. 


\section{Hybrid-Time Physics-Informed Neural Network}
\label{sec:method}


\subsection{Known physics}
\label{sec:physics}

Between jumps, a solution of $\mathcal H_\theta$ satisfies:
\begin{itemize}
  \setlength{\itemsep}{2pt}%
  \setlength{\parsep}{0pt}%
  \setlength{\parskip}{0pt}%
  \item[(P1)] \emph{Flow:} the mode $\sigma$ stays constant and
    $\dot\xi=f_\sigma(\xi;\theta)$.
  \item[(P2)] \emph{Flow condition:} $\gamma_{\sigma,i}(\xi;\theta)\ge 0$ for all $i\in I(\sigma)$.
\end{itemize}
At a jump, triggered by guard $i$ it satisfies:
\begin{itemize}
  \setlength{\itemsep}{2pt}%
  \setlength{\parsep}{0pt}%
  \setlength{\parskip}{0pt}%
  \item[(P3)] \emph{Jump condition:}  $\gamma_{\sigma,i}(\xi;\theta)\le 0$.
  \item[(P4)] \emph{Reset:} the jump changes the state abruptly from $(\xi,\sigma)$ to
    $\big(R_{\sigma,i}(\xi;\theta),\,\nu(\sigma,i)\big)$.
\end{itemize}
Rules (P2) and (P3) act together. A flow interval ends with all guards
nonnegative, and a jump needs the triggering guard to be nonpositive, 
so this guard is zero at the jump time. This is used to locate
the jumps. 

\subsection{Single network cannot learn hybrid systems}
\label{sec:surrogate}

A standard PINN represents the solution by a single neural network $\psi(t)$.
Such a network is a continuous function of time and cannot represent a jump.
We say that $\psi$ is \emph{$\varepsilon$-close} to a solution $\phi$ if every point of $\phi$ is within $\varepsilon$ of $\psi$ in both time and state: for every $(t,j)\in\operatorname{dom}\phi$ there is an $s\in[0,T]$ with $|s-t|\le\varepsilon$ and $|\psi(s)-\xi(t,j)|\le\varepsilon$. We say that $\phi$ \emph{jumps by $\rho$ at time $t^\star$} if its $j$-th jump, for some $j$, occurs at $t^\star$ and changes the physical variables by $\rho=|\xi(t^\star,j)-\xi(t^\star,j-1)|$.

\begin{proposition}[Accuracy versus physical consistency at a jump]
\label{prop:smooth}
Let a solution $\phi$ of $\mathcal H$ jump by $\rho>0$ at time $t^\star$, and let $\psi$ be a continuously differentiable function that is $\varepsilon$-close to $\phi$ for some $\varepsilon<\rho/2$. Then, for every mode $\sigma$, the physics residual of $\psi$ near the jump satisfies
\[
\int_{I_\varepsilon}\bigl|\dot\psi-f_\sigma(\psi)\bigr|^2dt \ge \frac{\big[\rho-2\varepsilon(1+M_\sigma)\big]_+^2}{2\varepsilon},
\]
where $I_\varepsilon$ is the set of $t\in[0,T]$ with $|t-t^\star|\le\varepsilon$, and $M_\sigma$ is the largest value of $|f_\sigma(\psi(t))|$ on $I_\varepsilon$.
\end{proposition}

\begin{IEEEproof}
Let $\xi^-$ and $\xi^+$ be the states just before and just after the jump, so that $|\xi^+-\xi^-|=\rho$. Because $\psi$ is $\varepsilon$-close to $\phi$, there are times $s_1$ and $s_2$, both within $\varepsilon$ of $t^\star$, with $|\psi(s_1)-\xi^-|\le\varepsilon$ and $|\psi(s_2)-\xi^+|\le\varepsilon$. By the triangle inequality, $|\psi(s_2)-\psi(s_1)|\ge\rho-2\varepsilon>0$, so $s_1\ne s_2$; also $|s_2-s_1|\le2\varepsilon$. Let $r=\dot\psi-f_\sigma(\psi)$. Integrating $\dot\psi=r+f_\sigma(\psi)$ from $s_1$ to $s_2$ gives $\rho-2\varepsilon\le\bigl|\int_{s_1}^{s_2}r\,dt\bigr|+2\varepsilon M_\sigma$, and, by the Cauchy--Schwarz inequality, $\bigl|\int_{s_1}^{s_2}r\,dt\bigr|^2\le2\varepsilon\int_{I_\varepsilon}|r|^2dt$.
\end{IEEEproof}

As $\varepsilon\to0$, the bound grows to infinity like $\rho^2/(2\varepsilon)$, while $M_\sigma$ stays bounded if $\psi$ does. So the more accurately a single smooth network follows a jump, the larger its physics residual near the jump, and hence its flow loss $L_{\mathrm{flow}}$ in Section~\ref{sec:loss}. HT-PINNs avoid this by using a separate network on each flow interval.

\subsection{One network per flow interval}
We split the time interval $[0,T]$ at the jump times and place one neural network on each
piece, so that each network only has to represent the flow in a single mode. The neural network involves the following quantities:
\begin{itemize}
  \setlength{\itemsep}{2pt}%
  \setlength{\parsep}{0pt}%
  \setlength{\parskip}{0pt}%
  \item[(1)] \emph{Known}: the time interval $[0,T]$, the output map $h$, the maps $f_\sigma$, $\gamma_{\sigma,i}$ and $R_{\sigma,i}$ up to the parameter $\theta$,
        and the transition map $\nu$.
  \item[(2)] \emph{Unknown, chosen by search}: the
        number of jumps $J$, the initial mode $\hat\sigma_0$, and the guards
        $i_1,\dots,i_J$, where $i_j$ is the guard that triggers jump~$j$. We show the selection of these parameters in Section~\ref{sec:structure}. 
  \item[(3)] \emph{Computed, not estimated}: the modes
        $\hat\sigma_1,\dots,\hat\sigma_J$ after the jumps, $\hat\sigma_j=\nu(\hat\sigma_{j-1},i_j)$, since jump~$j$ leaves mode $\hat\sigma_{j-1}$ through guard $i_j$.

  \item[(4)] \emph{Unknown, trained}: the jump times
        $\hat\tau_1,\dots,\hat\tau_J$, the neural networks for 
        $\hat\xi_0,\dots,\hat\xi_J$ (that is, their weights
        $w_0,\dots,w_J$), and the parameter $\theta$.
\end{itemize}
Together with the fixed endpoints $\hat\tau_0=0$ and $\hat\tau_{J+1}=T$,
the jump times are ordered,
$  0=\hat\tau_0<\hat\tau_1<\dots<\hat\tau_J<\hat\tau_{J+1}=T$,
and split $[0,T]$ into the $J+1$ pieces $[\hat\tau_j,\hat\tau_{j+1}]$,
$j=0,\dots,J$. Piece~$j$ covers the trajectory after exactly $j$ jumps:
its mode is $\hat\sigma_j$, and the network $\hat\xi_j$ approximates the
physical variables on it.

To evaluate $\hat\xi_j$ at a time $t$ in piece~$j$, we first convert $t$
into the rescaled time
\begin{equation}
  s=\frac{t-\hat\tau_j}{\hat\tau_{j+1}-\hat\tau_j}\in[0,1],
  \label{eq:rescaled}
\end{equation}
which is the fraction of piece~$j$ that has elapsed at time $t$.
As in switching-time optimization~\cite{XX-PJA:04}, this rescaling fixes the jump times at $s\in\{0,1\}$, so that they enter the loss as parameters.
We then evaluate a neural network $\mathcal N_{w_j}$ with weights
$w_j$ at $s$ as
$  \hat\xi_j(t)=\mathcal N_{w_j}(s)$.
Combining the pieces, the trajectory of the neural network is
\begin{equation}
  \hat\phi(t,j)=\big(\hat\xi_j(t),\hat\sigma_j\big),
  \qquad t\in[\hat\tau_j,\hat\tau_{j+1}],\quad j=0,\dots,J,
  \label{eq:surrogate}
\end{equation}
that is, after $j$ jumps the physical variables are given by $\hat\xi_j$
and the mode is $\hat\sigma_j$. At the jump time $\hat\tau_j$, piece $j-1$
ends and piece $j$ begins. We write
$  \hat\xi_j^-=\hat\xi_{j-1}(\hat\tau_j)=\mathcal N_{w_{j-1}}(1),
  \hat\xi_j^+=\hat\xi_j(\hat\tau_j)=\mathcal N_{w_j}(0)$,
for the physical variables just before and just after jump~$j$, that is,
the last value of the previous mode ($s=1$) and the first value of the
next mode ($s=0$). This is shown in Fig.~\ref{fig:loss}.

\providecommand{\rev}[1]{\textcolor{blue}{#1}}

\subsection{Physics-informed loss}
\label{sec:loss}

The loss penalizes deviations from the data and from the known physics of Section~\ref{sec:physics}:
\begin{equation}
\begin{aligned}
&L_{\mathrm{data}}  = \sum_{k=1}^{K}\bigl|h(\hat\xi_{\hat j_k}(t_k))-y_k\bigr|^2,\\
  & L_{\mathrm{flow}}  = \sum_{j=0}^{J}\int_{T_j}\bigl|\dot{\hat\xi}_j-f_j(\hat\xi_j)\bigr|^2\,dt,\\
&L_{\mathrm{guard}} = \sum_{j=0}^{J}\sum_{i\in I_j}\int_{T_j}\bigl[-\gamma_{j,i}(\hat\xi_j)\bigr]_+^2\,dt
  +\sum_{j=1}^{J}\bigl[\gamma_{j-1,i_j}(\hat\xi_j^-)\bigr]_+^2,\\
  & L_{\mathrm{reset}} = \sum_{j=1}^{J}\bigl|\hat\xi_j^+-R_{j-1,i_j}(\hat\xi_j^-)\bigr|^2,
\end{aligned}
\label{eq:loss-terms}
\end{equation}

\begin{figure}[t]
  \centering
  \resizebox{\linewidth}{!}{%
  \begin{tikzpicture}[
      >=Stealth,
      every node/.style={font=\footnotesize},
      piece/.style={thick, black},
      data/.style={circle, fill=black, draw=white, line width=0.3pt, inner sep=1pt},
      jump/.style={red},
    ]
    \draw[->] (0,0) -- (8.4,0) node[right] {$t$};
    \draw[->] (0,0) -- (0,2.15) node[left] {$\hat\xi$};
    \node[below] at (0,0) {$0$};
    \draw (8.0,0.06) -- (8.0,-0.06);
    \node[below] at (8.0,0) {$T$};
    \foreach \x/\j in {2.8/1, 5.6/2} {
      \draw[jump, dashed] (\x,0) -- (\x,1.95);
      \node[jump, below] at (\x,0) {$\hat\tau_{\j}$};
    }
    \foreach \a/\b/\j in {0.25/2.5/0, 3.1/5.3/1, 5.9/7.75/2}
      \draw[<->] (\a,-0.23) -- (\b,-0.23) node[midway, fill=white, inner sep=1pt] {$\Delta_{\j}$};
    \draw[piece] (0,1.45) .. controls (1.0,1.5) and (2.2,1.1) .. (2.8,0.6)
      node[pos=0.2, data] {} node[pos=0.5, data] {} node[pos=0.75, data] {};
    \draw[piece] (2.8,1.7) .. controls (3.7,1.78) and (4.95,1.25) .. (5.6,0.65)
      node[pos=0.2, data] {} node[pos=0.5, data] {} node[pos=0.78, data] {};
    \draw[piece] (5.6,1.5) .. controls (6.4,1.58) and (7.3,1.3) .. (8.0,1.0)
      node[pos=0.22, data] {} node[pos=0.5, data] {} node[pos=0.8, data] (yk) {};
    \node at (1.15,1.83) {$\mathcal N_{w_0}(s)$};
    \node at (3.95,2.0) {$\mathcal N_{w_1}(s)$};
    \node at (6.75,1.82) {$\mathcal N_{w_2}(s)$};
    \foreach \x/\ylo/\yhi/\j in {2.8/0.6/1.7/1, 5.6/0.65/1.5/2} {
      \draw[jump, ->] (\x+0.06,\ylo+0.07) to[bend right=40] node[right, inner sep=1pt] {$R$} (\x+0.06,\yhi-0.07);
      \draw[jump, fill=white] (\x,\ylo) circle (1.6pt);
      \fill[jump] (\x,\yhi) circle (1.6pt);
      \node[jump, anchor=east, inner sep=1pt] at (\x-0.1,\yhi) {$\hat\xi_{\j}^+$};
      \node[jump, anchor=north east, inner sep=1pt] at (\x-0.12,\ylo-0.06) {$\gamma(\hat\xi_{\j}^-)=0$};
    }
    \draw[dotted] (yk) -- (yk |- 0,0) node[above left, inner sep=1pt] {$t_k$};
  \end{tikzpicture}}
  \caption{Neural network for each region (black), and resets when guards become zero (red).}
  \label{fig:loss}
\end{figure}
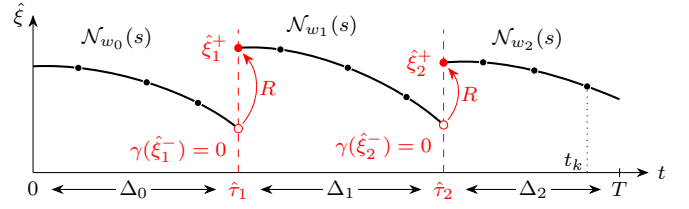
\noindent where $T_j=[\hat\tau_j,\hat\tau_{j+1}]$, $I_j=I(\hat\sigma_j)$, $f_j=f_{\hat\sigma_j}(\cdot;\theta)$,
$\gamma_{j,i}=\gamma_{\hat\sigma_j,i}(\cdot;\theta)$ and $R_{j,i}=R_{\hat\sigma_j,i}(\cdot;\theta)$.
The term $L_{\mathrm{data}}$ fits the measurements, $L_{\mathrm{flow}}$ enforces (P1), the
two parts of $L_{\mathrm{guard}}$ enforce (P2) and (P3), and $L_{\mathrm{reset}}$
enforces (P4). The total loss
\begin{equation}
L = L_{\mathrm{data}}+\lambda_1 L_{\mathrm{flow}}+\lambda_2 L_{\mathrm{guard}}
    +\lambda_3 L_{\mathrm{reset}}
\label{eq:loss}
\end{equation} 
is minimized over the weights $w_0,\dots,w_J$, the jump times
$\hat\tau_1,\dots,\hat\tau_J$, and $\theta\in\Theta$, where 
$\lambda_1,\lambda_2,\lambda_3>0$ are constants. The minimizing $\hat\theta$ is our
estimate of $\theta^\star$, and the minimizing $\hat\phi$ supplies the mode,
the jump times, and the unmeasured physical variables.
In practice, each residual in~\eqref{eq:loss-terms} is divided by a scale, the noise level for the data and a prescribed accuracy for the physics, and $L_{\mathrm{data}}$ and the integrals are divided by $K$ and $T$.

\textbf{Learning jump times.} 
With $\Delta_j=\hat\tau_{j+1}-\hat\tau_j$ and $\mathcal N_{w_j}'=d\mathcal N_{w_j}/ds$, the chain rule gives $\dot{\hat\xi}_j=\mathcal N_{w_j}'(s)/\Delta_j$. Therefore,
\begin{equation}
L_{\mathrm{flow}}=\sum_{j=0}^{J}\Delta_j\int_0^1\Big|\frac{\mathcal N_{w_j}'(s)}{\Delta_j}-f_j\big(\mathcal N_{w_j}(s)\big)\Big|^2ds,
\label{eq:flow-rescaled}
\end{equation}
where the jump times $\hat\tau_j$ appear in $\Delta_j$. Each piece $j\in\{0,\cdots,J-1\}$ must follow its flow, and end on its guard. A wrong piece length makes these requirements conflict: if $\hat\tau_j$ is too early, the piece $j-1$ would have to reach its guard faster than its flow allows and if it is too late, the piece $j-1$ would cross its guard before the jump, which $L_{\mathrm{guard}}$ penalizes. 

\subsection{Choosing the number of jumps and the guards}
\label{sec:structure}

The number of jumps $J$, the initial mode $\hat\sigma_0$ and the guards $i_1,\dots,i_J$ are discrete and cannot be trained, so Algorithm~\ref{alg:structure} enumerates them. A choice is admissible if each jump uses a guard of the current mode, $i_j\in I(\hat\sigma_{j-1})$, which gives the next mode $\hat\sigma_j=\nu(\hat\sigma_{j-1},i_j)$. For $J=0,1,\dots$ the neural network is trained for every admissible choice, and the algorithm returns the best choice for the first $J$ at which the loss falls below a tolerance $\varepsilon$; stopping at the smallest such $J$ avoids explaining the data with superfluous jumps. The piece lengths are parameterized as $\Delta=T\text{softmax}(\alpha)$, $\alpha\in \mathbb{R}^{J+1}$, which keeps the jump times inside $(0,T)$. We initialize the jump times by an exact least-squares fit of $J+1$ quadratic segments to the data, by dynamic programming, and train each choice from three starts, two with perturbed jump times. For each $J$, there are at most $|\Sigma|m^J$ admissible choices, where $m=\max_\sigma |I(\sigma)|$.

\begin{algorithm}[t]
  \caption{Hybrid-Time Physics-Informed Neural Network (HT-PINN)}
  \label{alg:structure}
  \KwIn{data $\{(t_k,y_k)\}_{k=1}^{K}$, maximum number of jumps $J_{\max}$, tolerance $\varepsilon$, iterations $M$}
  \KwOut{$J$, $\hat\sigma_0$, $i_1,\dots,i_J$, jump times $\hat\tau_1,\dots,\hat\tau_J$, neural networks $\mathcal N_{w_0},\dots,\mathcal N_{w_J}$, and $\hat\theta$}
  \For{$J=0,1,\dots,J_{\max}$}{
    \ForEach{admissible $(\hat\sigma_0,i_1,\dots,i_J)$}{
      initialize $\mathcal N_{w_0},\dots,\mathcal N_{w_J}$, $\theta$ and $\alpha\in\mathbb{R}^{J+1}$\;
      
      \For{$n=1,\dots,M$}{
      \!\!  $\Delta_j\!\leftarrow\! T\,\mathrm{softmax}(\alpha)_j$,  $\hat\tau_j\!\leftarrow\Delta_0\!+\!\dots\!+\!\Delta_{j\!-\!1}$\;
        
        gradient step on $w_0,\dots,w_J$, $\alpha$ and $\theta$\;
      }
      store the final $L$ with the trained networks, jump times and $\theta$\;
    }
    \If{the smallest $L$ stored for this $J$ is below $\varepsilon$}{
      \KwRet{that choice and its trained variables}\;
    }
  }
  \KwRet{the choice with the smallest $L$}\;
\end{algorithm}

\section{\texorpdfstring{{Numerical Experiments}}{Numerical Experiments}}
\label{sec:experiments}

\begin{figure*}[t]
  \centering
  \includegraphics[width=\textwidth]{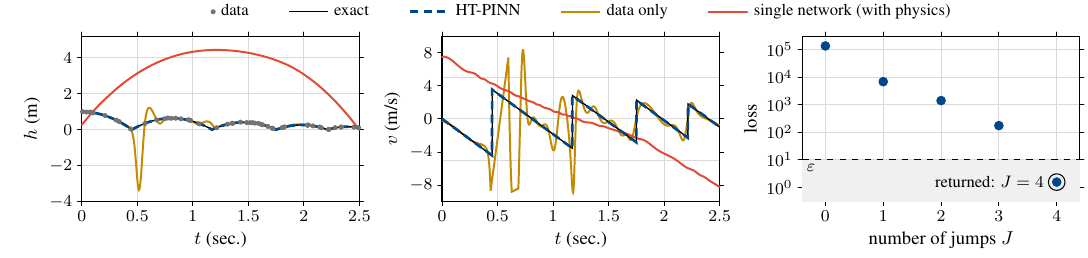}
  \vspace{-7mm}
\caption{Bouncing ball with height over time (left), unmeasured velocity over time (middle), smallest loss for number of jumps $J$ (right). The first $J$ below the tolerance $\varepsilon$ is $J=4$.}
  \label{fig:ball}
\end{figure*}
We test HT-PINNs on two examples. The first is a bouncing ball, whose exact solution is known. The second is a robotic disassembly task, in which a gantry robot pulls a press-fitted pin out of a housing, for which MuJoCo simulation is the reference.\footnote{Code and a video are available at \url{https://github.com/shivanshutripath/dapinn-hybrid-surrogates}.}
Both examples share the same architecture. Each of the $J+1$ pieces is represented by its own network $\mathcal N_{w_j}$ with two hidden layers of $32$ tanh neurons. Each piece has $64$ collocation points at fixed, evenly spaced values of the rescaled time $s$ in~\eqref{eq:rescaled}. We train with Adam ($2000$ steps, learning rate $5\times10^{-3}$), then L-BFGS; the search takes about $4$~min for the ball on one CPU core and $7$~min for the pin on seven.

\subsection{\texorpdfstring{{Bouncing ball}}{Bouncing ball}}
\label{sec:exp_ball}

We consider a hybrid system that models the height and velocity of a bouncing ball as $\xi=(h,v)$, the guard is $\gamma=h$, and the unknown parameter is the restitution coefficient $\theta=e\in\Theta=[0.3,1]$. In the form~\eqref{eq:hs},
\begin{equation}
\mathcal H_\theta^{\mathrm{ball}}:\
\begin{cases}
\dot h=v,\quad \dot v=-g, & h\ge0,\\
h^+=h,\quad v^+=-e\,v, & h\le0,
\end{cases}
\qquad y=h,
\label{eq:ball}
\end{equation}
so the ball falls freely while it is above the ground and, at each impact, its velocity is reversed and scaled by $e$.

In the true solution, the ball is dropped at rest from $h=1$~m, with $e^\star=0.8$ and $g=9.81$~m/s$^2$, and it hits the ground four times in $[0,T]$ with $T=2.5$~s, at $0.452$, $1.174$, $1.752$ and $2.214$~s. We measure height at $50$ random times with Gaussian noise of standard deviation $1$ cm. We use $\lambda_1=\lambda_2=\lambda_3=10$ and $J_{\max}=6$, initialize $e$ at $0.5$. Algorithm~\ref{alg:structure} finds the number of jumps as the smallest $J$ whose loss falls below $\varepsilon=10$. Each network outputs $(h,v)$, and only $h$ is measured.

We compare the proposed neural network model with the results of two different baseline models, each a single network $h(t)$ on $[0,T]$ with $v=\dot h$. The first, which we call the \emph{data-only} model, fits the data alone. The second, called the \emph{single network}, is a PINN that also enforces the physics, with each impact replaced by a stiff penalty contact, $\dot v=-g+k_p[-h]_+-c_pv$ for $h<0$, where $k_p=10^4$~s$^{-2}$ and $c_p$ gives the true $e$, but does not predict the jumps as in Proposition~\ref{prop:smooth}.

We observe that Algorithm~\ref{alg:structure} returns $J=4$ (Fig.~\ref{fig:ball}, right). Each added jump lowers the loss: the loss falls from $1.3\times 10^5$ at $J=0$ to $171$ at $J=3$, and to $1.5$ at $J=4$, which is below the threshold. The HT-PINN estimates $\hat e=0.8003$ and the impact times within $2$~ms; over three noise realizations, it finds $J=4$ each time and $\hat e=0.7994\pm0.0010$ (mean~$\pm$~s.d.). The data-only network fits the samples below the noise level but reaches the ground only once instead of four times, and there it dips more than $3$~m below the ground. The single network with physics cannot satisfy the data and the penalty flow at the same time, and its loss stays above $10^5$ (Fig.~\ref{fig:ball}).

\subsection{{{Robot disassembly in MuJoCo}}}
\label{sec:exp_cover}

\newsavebox{\pinmodes}
\begin{figure*}[t]
  \centering
  \sbox{\pinmodes}{\scalebox{0.9}{%
    \begin{tikzpicture}[font=\footnotesize, line width=0.6pt,
        >={Latex[length=1.8mm, width=1.3mm]},
        mode/.style={draw, rounded corners=1.5pt, align=center, inner sep=2.5pt, minimum width=29mm},
        lab/.style={font=\scriptsize, align=center, inner sep=2pt}]
      \node[mode] (S) at (0,0)
        {\underline{stick ($\mathsf S$)}\\[3pt] $\dot\delta=\dot U-v$\\ $m\dot v=0$};
      \node[mode] (K) at (0,-2.7)
        {\underline{slip ($\mathsf K$)}\\[3pt] $\dot\delta=\dot U-v$\\ $m\dot v=F-mg-\mu_kN$};
      \node[mode, minimum width=23mm, anchor=west] (R) at ([xshift=16mm]K.east)
        {\underline{released ($\mathsf R$)}\\[3pt] $\dot\delta=\dot U-v$\\ $m\dot v=F-mg$};
      \draw[->] ([xshift=-7mm]S.west) -- node[lab, above] {start} (S.west);
      \draw[->] ([xshift=-5mm]S.south) -- node[lab, left]
        {slipping\\ $F-mg\ge\mu_sN$} ([xshift=-5mm]K.north);
      \draw[->] ([xshift=5mm]K.north) -- node[lab, right, align=left]
        {stuck\\ $v\le0$ and $F-mg\le\mu_kN$} ([xshift=5mm]S.south);
      \draw[->] (K.east) -- node[lab, above] {released} node[lab, below] {$U-\delta\ge L$} (R.west);
    \end{tikzpicture}}}%
  {\footnotesize
  \begin{tabular}{@{}c@{\hspace{15mm}}c@{}}
    \begin{tikzpicture}[font=\scriptsize,
        tag/.style={fill=white, draw=black, text=black, line width=0.3pt, rounded corners=1pt, inner sep=1.5pt},
        pointer/.style={white, -{Latex[length=1.6mm, width=1.2mm]}, line width=0.6pt}]
      \node[inner sep=0, anchor=south west] (img) at (0,0)
        {\includegraphics[height=\ht\pinmodes]{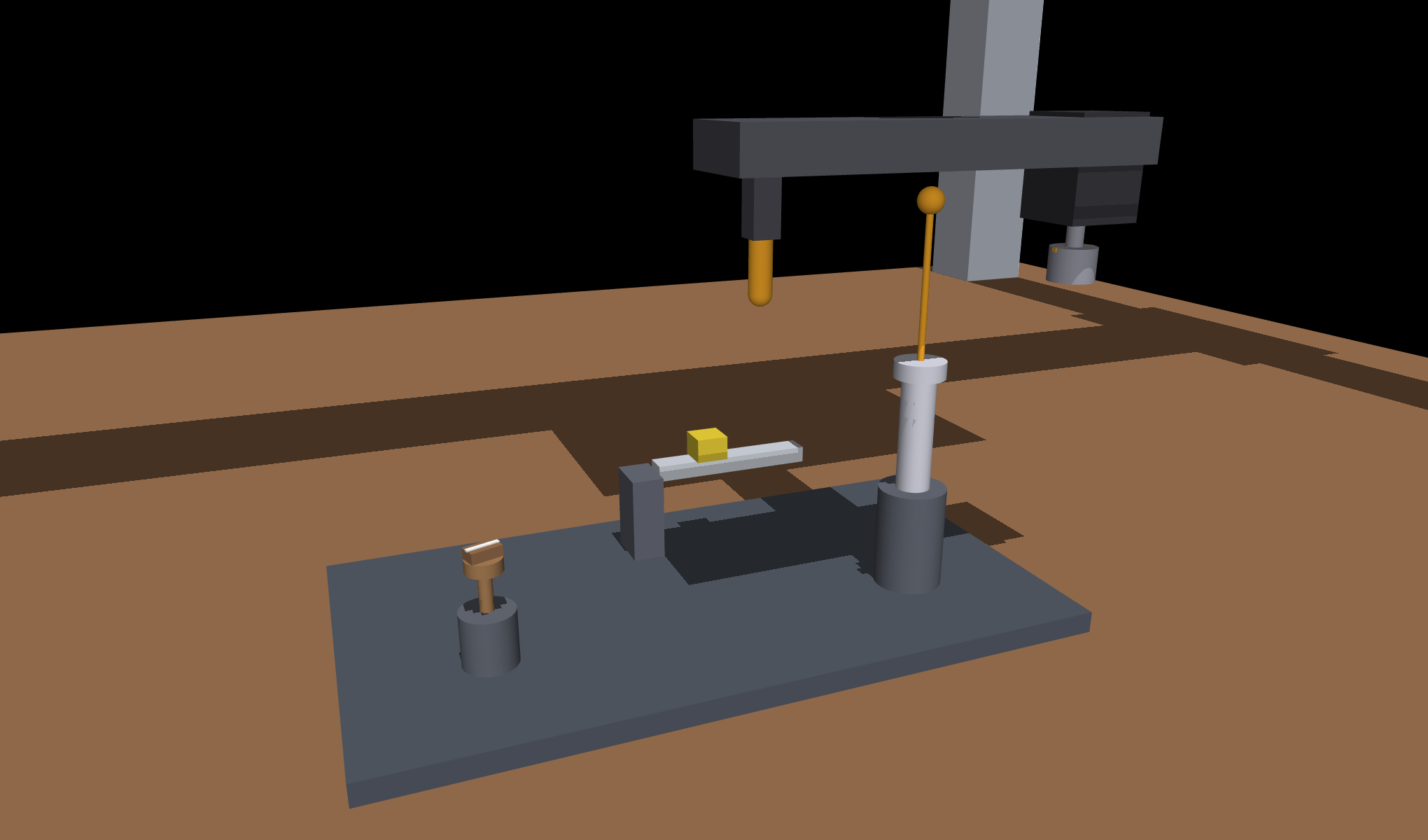}};
      \begin{scope}[x={(img.south east)}, y={(img.north west)}]
        \draw[pointer] (0.30,0.86) node[tag, anchor=east] {gantry robot} -- (0.49,0.83);
        \draw[pointer] (0.84,0.58) node[tag, anchor=west] {spring} -- (0.652,0.62);
        \draw[pointer] (0.84,0.45) node[tag, anchor=west] {pin} -- (0.66,0.50);
        \draw[pointer] (0.84,0.30) node[tag, anchor=west] {housing} -- (0.66,0.35);
        \draw[pointer] (0.36,0.60) node[tag, anchor=east] {latch} -- (0.47,0.47);
        \draw[pointer] (0.14,0.42) node[tag, anchor=east] {screw} -- (0.32,0.34);
      \end{scope}
    \end{tikzpicture}
    & \usebox{\pinmodes}\\[1mm]
    (a) & (b)
  \end{tabular}}
  \caption{{Extraction of a press-fitted pin. (a) The MuJoCo cell: the gantry robot pulls the pin out of the housing through the spring of its series-elastic puller. (b) The hybrid system~\eqref{eq:pin}: the flow of each mode, with $F=k\delta+c\,(\dot U-v)$, and the jumps with the conditions that trigger them. The resets are the identity, and the measured output is $y=\delta$.}}
  \label{fig:pin_setting}
\end{figure*}

We consider a gantry robot that pulls a press-fitted steel pin out of a housing using a spring (Fig.~\ref{fig:pin_setting}a). The pin has mass $m=50$~g and is held over a length $L=19.5$~mm by an interference fit that presses on it with a normal force $N=40$~N, and its coefficients of static and kinetic friction are $\mu_s=0.45$ and $\mu_k=0.30$.  The robot pulls the pin using a spring of stiffness $k=2000$~N/m and damping $c=2$~N\,s/m. The robot is commanded to raise the hook at $20$~mm/s. Because the static friction exceeds the kinetic friction, the pin sticks while the spring stretches. It slips only when the pulling force reaches the static friction, and sticks again when it stops. The number of slips depends on the friction of the returned unit and is not known in advance.

The reference is a MuJoCo simulation of the robot cell with a time step of $0.25$ ms. MuJoCo limits the friction force to $\mu_sN$ while the pin is at rest, to $\mu_kN$ while it slides, and to zero once it has left the fit. In our experiments, the pin breaks away four times at $18.5$~N and the fourth slip carries the pin out of the fit. We consider 3 modes as stuck (S), slipping (K), and released from fit (R), so the sequence of modes is SKSKSKSKR. 

We measure the deflection $\delta$ of the spring, which is the difference between the current length and the unstressed length of the spring. Let $x$ and $v$ be the lift and the velocity of the pin. We take as physical variables $\xi=(\delta,v)$, with the deflection $\delta=U-x$ of the spring, so that the pulling force is $F=k\delta+c(\dot U-v)$. With the time $t$ appended to $\xi$ ($\dot t=1$), the system in the form~\eqref{eq:hs} is given as:
\begin{equation}
\mathcal H_\theta^{\mathrm{pin}}:\
\begin{cases}
\dot\delta=\dot U(t)-v,\quad m\dot v=a_\sigma, & \gamma_{\sigma,i}\ge0,\\[1mm]
\xi^+=\xi,\quad \sigma^+=\nu(\sigma,i), & \gamma_{\sigma,i}\le0,
\end{cases}
\qquad y=\delta,
\label{eq:pin}
\end{equation}
with the net forces $a_{\mathsf S}=0$, $a_{\mathsf K}=F-mg-\mu_kN$ and $a_{\mathsf R}=F-mg$, and the guards and transitions are given as (Fig.~\ref{fig:pin_setting}b):
\[
\begin{array}{lc}
{\mathsf S\to\mathsf K:} & \gamma_{\mathsf S,1}=\mu_sN-(F-mg),\\
{\mathsf K\to\mathsf S:} & ~~~~~~~~~~~\gamma_{\mathsf K,1}=\max\{v,\ F-mg-\mu_kN\} ,\\
{ \mathsf K\to\mathsf R:} &\!\!\!\! \gamma_{\mathsf K,2}=L-(U(t)-\delta) .
\end{array}
\]
A slip returns to stuck only when the pin has stopped and is still being slowed by friction, that is, when $v=0$ and $F-mg\le\mu_kN$ hold together. The arrest guard is the larger of these two quantities, so it is nonpositive only when both are. The same slip can instead carry the pin out of the fit: its lift is $U(t)-\delta$, and it exits when this lift reaches the fit length $L$. After the exit the pin is free, so mode $\mathsf R$ has no further event. 

The hook motion $U(t)$, recorded from the robot, is known. The unknowns are the friction coefficients $\theta=(\mu_s,\mu_k)\in\Theta=[0.2,0.8]\times[0.1,0.6]$. We measure the deflection $\delta$ for $T=1.48$~s from the start of the pull, sampled at $1$~kHz with Gaussian noise of standard deviation $25~\mu$m. We use $\lambda_1=\lambda_2=\lambda_3=1$, $J_{\max}=9$, $\varepsilon=10$ and $\hat\sigma_0=\mathsf S$, since the pin is at rest at the start, and initialize $\theta$ at $(0.37,0.24)$.

\begin{figure*}[t]
  \centering
  \includegraphics[width=\textwidth]{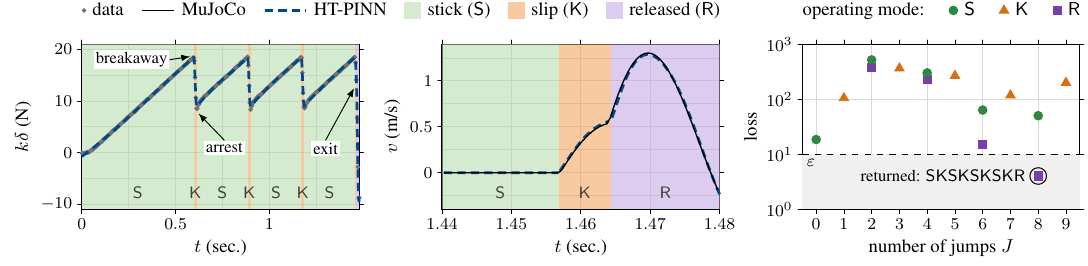}
  \vspace{-4mm}
  \caption{{Extraction of a press-fitted pin in MuJoCo. Left: measured spring force $k\delta$; the shading shows the modes of the HT-PINN, stick (green), slip (orange) and released (purple). Middle: unmeasured velocity of the pin during the last slip and the exit. Right: smallest loss for each $J$, shaded below the tolerance $\varepsilon$.}}
  \label{fig:pin}
\end{figure*}

\begin{figure}[t]
  \centering
  \includegraphics[width=0.5\textwidth]{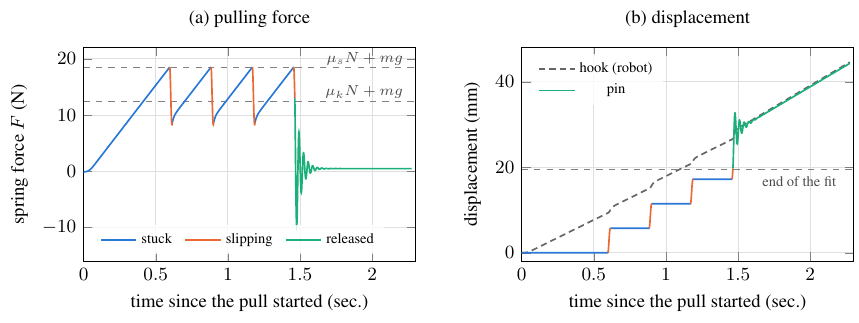}
  \caption{Force and displacement during the extraction of the press-fitted pin in MuJoCo, colored by the state of the pin: stuck ($\mathsf S$), slipping ($\mathsf K$) and released ($\mathsf R$). (a) Spring force of the puller; the dashed lines are the static and kinetic friction levels, offset by the weight of the pin. (b) Displacement of the hook, which the robot raises at a constant $20$~mm/s, and of the pin. }
  \label{fig:pin}
\end{figure}

\begin{table*}[t]
  \centering
  {
  \caption{The table shows the event times and friction coefficients ($\mu_s,\mu_k$) for MuJoCo simulator (true values) and HT-PINNs (our approach). Each $\tau_j$ is marked with the mode it enters: $\mathsf S$ (stick), $\mathsf K$ (slip), $\mathsf R$ (released).}
  \label{tab:pin}
  \footnotesize
  \setlength{\tabcolsep}{3.5pt}
  \begin{tabular}{lcccccccccc}\hline
  & $\tau_1$ ($\mathsf K$) & $\tau_2$ ($\mathsf S$) & $\tau_3$ ($\mathsf K$) & $\tau_4$ ($\mathsf S$) & $\tau_5$ ($\mathsf K$) & $\tau_6$ ($\mathsf S$) & $\tau_7$ ($\mathsf K$) & $\tau_8$ ($\mathsf R$) & $\mu_s$ & $\mu_k$\\ \hline
  MuJoCo & 0.5975 & 0.6145 & 0.8840 & 0.9010 & 1.1705 & 1.1875 & 1.4570 & 1.4640 & 0.4500 & 0.3000\\
  HT-PINN & 0.5967 & 0.6141 & 0.8845 & 0.9008 & 1.1697 & 1.1864 & 1.4568 & 1.4644 & 0.4495 & 0.3052\\ \hline
  \end{tabular}
  }
\end{table*}

Algorithm~\ref{alg:structure} returns $J=8$ jumps, the same as in the simulation as shown in Fig.~\ref{fig:pin} and Table~\ref{tab:pin}. At every even $J$, the candidate whose last slip ends in an exit has a lower loss than the one that ends in an arrest, and $J=8$ is the first $J$ whose loss, $4.1$, is below the tolerance as shown in Fig.~\ref{fig:pin} (right). The loss does not decrease steadily with $J$, because the candidates in between assume a wrong sequence of events. The single network, $J=0$, has a loss of only $18.6$ but estimates $\mu_s=0.75$: its $64$ collocation points, $23$~ms apart, do not resolve the $17$-ms slips. The static and kinetic friction coefficients are estimated within $0.1\%$ and $1.7\%$, and all eight event times within $1.1$~ms. The relative $L^2$ error is $1.1\%$ for the measured deflection and $7.5\%$ for the unmeasured velocity of the pin, shown in Fig.~\ref{fig:pin} (middle).
The estimated $\mu_s$ also gives the force $\mu_sN+mg$ needed to break the pin loose, which is what the robot needs to plan the extraction.

\section{Conclusion}
\label{sec:conc}
In this letter we study the modeling of hybrid systems whose mode and jump times are not measured. We propose HT-PINNs, which assign one physics-informed network to each flow interval, train the networks jointly with the jump times and the parameters, and select the number of jumps and the sequence of modes by enumeration. Further, we show that a single smooth network cannot follow a jump closely without a large physics residual near it. Numerical examples on a bouncing ball and on pin extraction in MuJoCo demonstrate that HT-PINNs recover the number of jumps and the parameters, whereas a single PINN and data-driven model cannot capture the jumps.
Several research directions remain open. These include validating HT-PINNs on a physical robot, and reducing the enumeration, which can grow exponentially with the number of jumps.


\bibliographystyle{IEEEtran}
\bibliography{refs}

\end{document}